\documentclass[conference,letterpaper]{IEEEtran}

\usepackage[utf8]{inputenc}
\usepackage[T1]{fontenc}
\usepackage{url}
\usepackage{ifthen}
\usepackage{cite}
\usepackage[cmex10]{amsmath} 

\usepackage{color}
\usepackage{amsfonts,amsmath,amssymb,amsthm}
\usepackage{latexsym,amscd}
\usepackage{amsbsy}
\usepackage{graphicx,multirow,bm}
\usepackage{algorithm}
\usepackage{algorithmicx}
\usepackage{algpseudocode}
\usepackage{xcolor}
\usepackage{tikz}
\usepackage{diagbox}
\usepackage{multicol}
\usepackage{arydshln}
\usepackage{booktabs}
\usepackage{times}
\usepackage{bm}
\usepackage{amsmath}
\usepackage{cases}
\usepackage{subfigure}
\usepackage{pgf}
\usepackage{tikz}
\usepackage{pgfplots}
\usepackage{scalefnt}

\newtheorem{thm}{Theorem}
\newtheorem{lem}{Lemma}

\newtheorem{defn}{Definition}

\newtheorem{cons}{Construction}

\newcommand{\f}{{\mathbb F}}

\newcommand{\fq}{{\mathbb F}_{q}}

\newcommand{\e}{{\epsilon}}
\usetikzlibrary{matrix,chains,positioning,arrows,calc,decorations,plotmarks,patterns,trees,fit,backgrounds,shapes.multipart,topaths}

\usetikzlibrary{external}

\pgfplotsset{compat=1.3}
\tikzstyle{help lines}=[black!20,dashed]

\begin{document}
\title{Construction of Locally Repairable Array Codes with Optimal Repair Bandwidth under the Rack-Aware Storage Model}
	


\author{%
  \IEEEauthorblockN{Yumeng Yang, Han Cai and Xiaohu Tang}
  \IEEEauthorblockA{ Information Security and National Computing Grid Laboratory\\
                    Southwest Jiaotong University\\
                    Chengdu, China\\
                    Email: yangyumeng@my.swjtu.edu.cn, hancai@swjtu.edu.cn, xhutang@swjtu.edu.cn}          
}

\maketitle

\begin{abstract}
 In this paper, we discuss codes for distributed storage systems with hierarchical repair properties. Specifically, we devote attention to the repair problem of the rack-aware storage model with locality, aiming to enhance the system's ability to repair a small number of erasures within each rack by locality and efficiently handling a rack erasure with a small repair bandwidth. By employing the regenerating coding technique, we construct a family of array codes with $(r,u-r+1)$-locality, where the $u$ nodes of each repair set are systematically organized into a rack. When the number of failures is less than $u - r + 1$, these failures can be repaired without counting the system bandwidth. In cases where the number of failures exceeds the locality, the failed nodes within a single rack can be recovered with optimal cross-rack bandwidth.
\end{abstract}

\section{Introduction}

With the development of the information technology and artificial intelligence, the question of how to store data has become increasingly crucial. Erasure coding has been introduced to address the challenges of large-scale storage in distributed systems, offering high fault-tolerance capabilities with significantly less redundancy compared to traditional replication methods, such as the well-known {\it Maximum Distance Separable}  (MDS) codes, which provide optimal failure tolerance and support for minimal storage overhead. However, when node failures occur in this storage system, traditional erasure codes suffer from challenges related to excessive bandwidth requirements for recovering the failed nodes.

To minimize data transmission during the repair process, Dimakis et al. \cite{DGW+10} introduced {\it regenerating codes}, aiming for an optimal tradeoff between storage and bandwidth. This approach involves contacting more surviving nodes but downloading only partial content from each, leading to a significant improvement compared to traditional schemes that necessitate downloading full data from helper nodes, referring to \cite{DGW+10,li2018generic,li2022pmds,holzbaur2020partial,guruswami2017repairing,CB19,MBW18,CV21,BBD+22,RSK+11,YB17,YB19,ZW19} as examples.
Simultaneously, another strategy to enhance repair efficiency is the utilization of {\it locally repairable codes}, which decrease the number of nodes accessed by the repair schemes to simplify the repair process, referring to \cite{gopalan2012locality,prakash2012optimal,cai2020optimal,cai2018optimal,chen2020improved,wang2014repair,xing2019construction,rawat2013optimal,TaBaFr16bounds,jin2019construction,mazumdar2014update,zeh2016bounds} as examples. In general, the reduction in the number of connected nodes results in a substantial decrease in repair bandwidth.

For a locally repairable code with $(r, \delta)$-locality, if the number of failed nodes is less than $\delta$, these failures can be easily recovered by leveraging the MDS property within the repair set. However, if $\delta$ or more failures occur simultaneously in a repair set, the locality breaks down.
Recently, Cai {\it et al.} \cite{CMS+23} combined locally repairable codes and regenerating codes to address erasure patterns where the number of failures exceeds the locality. They considered a practical scenario in which nodes in each repair set may be located in adjacent physical positions, forming what can be seen as a rack. Following the bandwidth assumption of the rack-aware model,  the cut-set bound is established for the rack-aware system with locality.

In this paper, we consider locally repairable codes under the rack-aware storage model, where each local repair set is organized into a rack. Inspired by regenerating codes, we generalize the Tamo-Barg code into an array form, resulting in a family of locally repairable array codes. The proposed code accommodates any number of failures in a single rack and achieves the optimal repair bandwidth for erasures beyond the local repair property. Furthermore, in comparison to traditional rack-based codes, this construction can support a broader range of erasure patterns with small repair bandwidth.

\section{Preliminaries}\label{sec: pre}

First of all,  we introduce some notation and repair models involved in this paper.
For any positive integers $a<b$, denote by $[a]$ the set $\{1,2,\cdots,a\}$ and $[a,b]$ the set $\{a,\cdots,b\}$. Let $q$ be a prime power and $\fq$ be a finite field of $q$ elements.
Denote by $\fq[x]$ the ring of polynomials over $\fq$. We represent ${\bm f}({\bm X})=(f^{(0)}(x_0),\cdots,f^{(l-1)}(x_{l-1}))^\top$ as an $l$-length vector consisting of polynomials, where ${\bm f}=(f^{(0)},f^{(1)},\ldots,f^{(l-1)})^\top$, $f^{(i)}\in \fq[x]$ for $0\leq i\leq l-1$ and  $\bm{X}=(x_0,\cdots,x_{l-1})^\top\in \fq^l$.

In this paper, we consider the array form of codes with locality. Herein, we give the formal definition.
\begin{defn}[Locally Repairable Array Code, \cite{gopalan2012locality, prakash2012optimal}]\label{def: array_lrc}
For $i\in[n]$, the $i$-th code symbol of an $(n,k;l)$ linear array code $\mathcal{C}$ has $(r,\delta)$ locality if there exists a subset $S_{i}\subseteq [n]$ (repair set) such that
\begin{enumerate}
\item[(1)] $i\in S_{i}$ and $|S_{i}|\leq r+\delta-1$;
\item[(2)] the minimum distance of the punctured code $\mathcal{C}|_{S_{i}}$ is at least $\delta$.
\end{enumerate}
\end{defn}
Furthermore, the code $\mathcal{C}$ is said to be a locally repairable code with $(r,\delta)$-locality if all the code symbols have $ (r, \delta)$-locality.

We employ Reed-Solomon (RS) codes as the building block of MDS array codes to design locally repairable array codes, whose definition is given below.

\begin{defn}[Generalized Reed-Solomon Code, \cite{MS77}]\label{Def_RS}
Let $A=\{\alpha_{1},\cdots,\alpha_{n}\}$ be the set of distinct elements over $\fq$. A Generalized Reed-Solomon code of length $n$ and dimension $k$ with evaluation points $A$ is defined as
\begin{equation*}
\begin{split}
{\rm GRS}(n,k,A,\boldsymbol{\nu})=&\{(\nu_{1}f(\alpha_{1}),\cdots,\nu_{n}f(\alpha_{n})):\\
&\hspace{2cm} f\in\fq[x],deg(f)<k\},
\end{split}
\end{equation*}
where $ \boldsymbol{\nu}=(\nu_{1},\cdots,\nu_{n})\in(\fq^{*})^n$.
\end{defn}
When $ \boldsymbol{\nu}=(1,1,\cdots,1)$, this code is called a Reed-Solomon code denoted as ${\rm RS}(n,k,A)$ for short.

\begin{defn}[Dual Code, \cite{MS77}]\label{def:dual}
Let $\mathcal{C}$ be a linear code over $\fq$ of length $n$. The dual code $\mathcal{C}^{\bot}$ of  code $\mathcal{C}$  is defined as
$$\mathcal{C}^{\bot}=\{(c'_{1},c'_{2},\cdots,c'_{n}):\sum_{i=1}^{n}c'_{i}c_{i}=0,\,\,\,\,\forall (c_{1},\cdots,c_{n})\in\mathcal{C}\}.$$
\end{defn}

The dual of an RS code is a GRS code. Precisely, $({\rm RS}(n,k,A))^{\bot}={\rm GRS}(n,n-k, A, \boldsymbol{\nu})$, where $\nu_{i}=\Pi_{j\neq i}(\alpha_{i}-\alpha_{j})^{-1}$, $i\in[n]$. Therefore, this RS code can also be defined by the parity-check equations
\begin{equation*}
\sum_{i=1}^{n}\alpha_{i}^{l}\nu_{i}c_{i}=0, \,\,\,\,\,l\in[0,n-k-1].
\end{equation*}
Since the constant multiplier $\nu_{i}$ is well-defined by the evaluation points $A$, we omit it for simplicity in the subsequent discussion.

\subsection{System Models}
We illustrate the connections and differences in storage models involved in this work. To begin with, we introduce two common storage models.

 Assume that a file $M$ is divided into $k$ blocks, encoded into $n$ blocks, and then placed in $n$ different storage nodes. 
\begin{itemize}
\item[$\bullet$]Homogeneous storage model
\end{itemize}

For the research of distributed storage codes, most studies have concentrated on the {\it homogeneous storage model}, in which nodes are distributed uniformly in different locations. As shown in Fig. $1$, if a failure occurs in the node $c_1$, $k\leq d<n$ surviving nodes (referred to as helper nodes) send $\beta_{1},\cdots,\beta_{d}$ symbols to a new replacement node respectively to repair the failed symbol. The repair bandwidth counts as the data transmission from all the helper nodes to the replacement node, i.e., $b=\sum_{i=1}^{d}\beta_{i}$.

\vspace{-0.3cm}

\tikzset{every picture/.style={line width=0.75pt}} 
\begin{figure}[ht]\label{fig: homo_model}\begin{center}
\begin{tikzpicture}[x=0.5pt,y=0.5pt,yscale=-1,xscale=1]

\draw   (79.58,517.22) -- (130.16,517.22) -- (130.16,544.84) -- (79.58,544.84) -- cycle ;
\draw   (79.58,557.22) -- (130.16,557.22) -- (130.16,584.84) -- (79.58,584.84) -- cycle ;
\draw   (79.58,596.22) -- (130.16,596.22) -- (130.16,623.84) -- (79.58,623.84) -- cycle ;
\draw   (79.58,655.22) -- (130.16,655.22) -- (130.16,682.84) -- (79.58,682.84) -- cycle ;
\draw [color={rgb, 255:red, 208; green, 2; blue, 27 }  ,draw opacity=1 ]   (79.58,517.22) -- (130.16,544.84) ;
\draw [color={rgb, 255:red, 208; green, 2; blue, 27 }  ,draw opacity=1 ]   (79.58,544.84) -- (130.16,517.22) ;
\draw   (194.58,579.22) -- (245.16,579.22) -- (245.16,606.84) -- (194.58,606.84) -- cycle ;
\draw    (130,570.5) .. controls (169.98,553.92) and (202.35,561.12) .. (219.7,577.24) ;
\draw [shift={(221,578.5)}, rotate = 225] [fill={rgb, 255:red, 0; green, 0; blue, 0 }  ][line width=0.08]  [draw opacity=0] (12,-3) -- (0,0) -- (12,3) -- cycle    ;
\draw    (130,609.5) .. controls (170.53,631.7) and (202.68,616.64) .. (219.26,608.37) ;
\draw [shift={(221,607.5)}, rotate = 153.43] [fill={rgb, 255:red, 0; green, 0; blue, 0 }  ][line width=0.08]  [draw opacity=0] (12,-3) -- (0,0) -- (12,3) -- cycle    ;
\draw    (130,667.5) .. controls (162.51,678.34) and (213.44,645.51) .. (220.7,609.16) ;
\draw [shift={(221,607.5)}, rotate = 99.21] [fill={rgb, 255:red, 0; green, 0; blue, 0 }  ][line width=0.08]  [draw opacity=0] (12,-3) -- (0,0) -- (12,3) -- cycle    ;
\draw    (180,548.5) .. controls (167,578.5) and (189,646.5) .. (229,645.5) ;

\draw (98.61,523.72) node [anchor=north west][inner sep=0.75pt]  [font=\footnotesize] [align=left] {$\displaystyle c_{1}$};
\draw (98.61,563.72) node [anchor=north west][inner sep=0.75pt]  [font=\footnotesize] [align=left] {$\displaystyle c_{2}$};
\draw (98.61,600.72) node [anchor=north west][inner sep=0.75pt]  [font=\footnotesize] [align=left] {$\displaystyle c_{3}$};
\draw (98.61,661.72) node [anchor=north west][inner sep=0.75pt]  [font=\footnotesize] [align=left] {$\displaystyle c_{n}$};
\draw (112.53,629.21) node [anchor=north west][inner sep=0.75pt]  [font=\large,rotate=-90.04]  {$...$};
\draw (213.61,585.72) node [anchor=north west][inner sep=0.75pt]  [font=\footnotesize] [align=left] {$\displaystyle c_{1}$};
\draw (221.61,629.72) node [anchor=north west][inner sep=0.75pt]  [font=\footnotesize] [align=left] {$\displaystyle d$};
\draw (143.61,546) node [anchor=north west][inner sep=0.75pt]  [font=\footnotesize] [align=left] {$\displaystyle \beta _{1}$};
\draw (150.61,604) node [anchor=north west][inner sep=0.75pt]  [font=\footnotesize] [align=left] {$\displaystyle \beta _{2}$};
\draw (156.61,646) node [anchor=north west][inner sep=0.75pt]  [font=\footnotesize] [align=left] {$\displaystyle \beta _{d}$};
\draw (181.31,627.21) node [anchor=north west][inner sep=0.75pt]  [font=\large,rotate=-90.04]  {$...$};
\end{tikzpicture}
\vspace{-0.3cm}
\caption{ Homogeneous Storage Model}\end{center}
\end{figure}
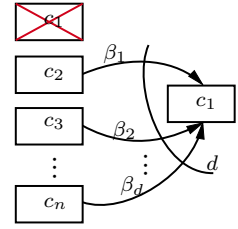

\begin{itemize}
\item[$\bullet$]Rack-aware storage model
\end{itemize}

In this model, $n$ storage nodes are divided into $\bar{n}$ groups of size $u$ and distributed to $\bar{n}$ different racks (refer to Fig. $2$).  Since nodes within the same rack have a significantly lower transmission cost than the inter-rack transmission, the intra-rack transmission is considered to be free. Assume that the node $c_{1,1}$ located in Rack $1$ fails, the nodes in $\bar{d}<\bar{n}$ helper racks first send data to a special relayer in each rack, then those relayers respectively transmit $\bar{\beta}_{1},\cdots,\bar{\beta}_{\bar{d}}$ symbols to the failed rack, and the remaining $u-1$ nodes in the failed rack transmit $\alpha_{2},\cdots,\alpha_{u}$ symbols to the new replacement node. Since the storage model only counts the data transmission between the racks, the repair bandwidth is equal to $\bar{b}=\sum_{i=1}^{\bar{d}}\bar{\beta}_{i}$.

\tikzset{every picture/.style={line width=0.75pt}} 
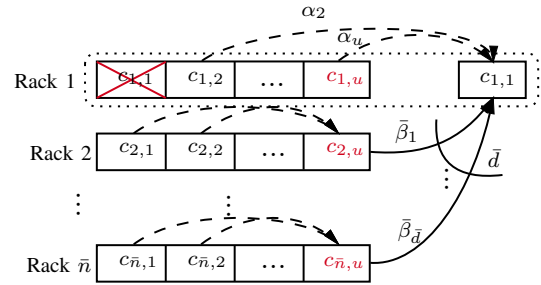
\begin{figure}[ht]\label{fig: rack_model}\begin{center}

\begin{tikzpicture}[x=0.5pt,y=0.5pt,yscale=-1,xscale=1]

\draw   (106.11,47.57) -- (312.58,47.57) -- (312.58,75.19) -- (106.11,75.19) -- cycle ;
\draw    (158.25,47.57) -- (158.25,74.85) ;
\draw    (210.39,47.57) -- (210.39,74.85) ;
\draw    (262.54,47.57) -- (262.54,74.85) ;
\draw   (106.11,102.14) -- (312.58,102.14) -- (312.58,129.88) -- (106.11,129.88) -- cycle ;
\draw    (158.25,102.14) -- (158.25,129.42) ;
\draw    (210.39,102.14) -- (210.39,129.42) ;
\draw    (262.54,102.14) -- (262.54,129.42) ;
\draw   (106.11,186.65) -- (312.58,186.65) -- (312.58,214.27) -- (106.11,214.27) -- cycle ;
\draw    (158.25,186.65) -- (158.25,213.94) ;
\draw    (210.39,186.65) -- (210.39,213.94) ;
\draw    (262.54,186.65) -- (262.54,213.94) ;
\draw [color={rgb, 255:red, 208; green, 2; blue, 27 }  ,draw opacity=1 ]   (106.11,75.19) -- (158.25,47.9) ;
\draw [color={rgb, 255:red, 208; green, 2; blue, 27 }  ,draw opacity=1 ]   (106.11,47.9) -- (158.25,75.19) ;
\draw   (380,47.65) -- (430.58,47.65) -- (430.58,75.27) -- (380,75.27) -- cycle ;
\draw  [dash pattern={on 4.5pt off 4.5pt}]  (289,47) .. controls (328.2,17.6) and (381.81,24.7) .. (404.64,46.64) ;
\draw [shift={(406,48)}, rotate = 226.27] [fill={rgb, 255:red, 0; green, 0; blue, 0 }  ][line width=0.08]  [draw opacity=0] (12,-3) -- (0,0) -- (12,3) -- cycle    ;
\draw  [dash pattern={on 4.5pt off 4.5pt}]  (185.58,47.14) .. controls (224.98,17.59) and (401.59,12.16) .. (405.92,46.41) ;
\draw [shift={(406,48)}, rotate = 271.59] [fill={rgb, 255:red, 0; green, 0; blue, 0 }  ][line width=0.08]  [draw opacity=0] (12,-3) -- (0,0) -- (12,3) -- cycle    ;
\draw    (313,116) .. controls (353.39,118.96) and (381.16,109.3) .. (404.92,76.51) ;
\draw [shift={(406,75)}, rotate = 125.22] [fill={rgb, 255:red, 0; green, 0; blue, 0 }  ][line width=0.08]  [draw opacity=0] (12,-3) -- (0,0) -- (12,3) -- cycle    ;
\draw    (313,200) .. controls (353.8,202.99) and (391.62,147.56) .. (405.79,76.08) ;
\draw [shift={(406,75)}, rotate = 101] [fill={rgb, 255:red, 0; green, 0; blue, 0 }  ][line width=0.08]  [draw opacity=0] (12,-3) -- (0,0) -- (12,3) -- cycle    ;
\draw  [dash pattern={on 4.5pt off 4.5pt}]  (134,102) .. controls (173.2,72.6) and (263.3,78.73) .. (287.6,100.64) ;
\draw [shift={(289,102)}, rotate = 226.27] [fill={rgb, 255:red, 0; green, 0; blue, 0 }  ][line width=0.08]  [draw opacity=0] (12,-3) -- (0,0) -- (12,3) -- cycle    ;
\draw  [dash pattern={on 4.5pt off 4.5pt}]  (185,102) .. controls (224.2,72.6) and (265.32,78.73) .. (287.66,100.64) ;
\draw [shift={(289,102)}, rotate = 226.27] [fill={rgb, 255:red, 0; green, 0; blue, 0 }  ][line width=0.08]  [draw opacity=0] (12,-3) -- (0,0) -- (12,3) -- cycle    ;
\draw  [dash pattern={on 4.5pt off 4.5pt}]  (185,186) .. controls (224.2,156.6) and (265.32,162.73) .. (287.66,184.64) ;
\draw [shift={(289,186)}, rotate = 226.27] [fill={rgb, 255:red, 0; green, 0; blue, 0 }  ][line width=0.08]  [draw opacity=0] (12,-3) -- (0,0) -- (12,3) -- cycle    ;
\draw  [dash pattern={on 4.5pt off 4.5pt}]  (134,186) .. controls (173.2,156.6) and (263.3,162.73) .. (287.6,184.64) ;
\draw [shift={(289,186)}, rotate = 226.27] [fill={rgb, 255:red, 0; green, 0; blue, 0 }  ][line width=0.08]  [draw opacity=0] (12,-3) -- (0,0) -- (12,3) -- cycle    ;
\draw    (363,91.5) .. controls (361,123.5) and (378,137.5) .. (413,133.5) ;
\draw  [dash pattern={on 0.84pt off 2.51pt}] (97.18,49.55) .. controls (97.18,45.13) and (100.76,41.55) .. (105.18,41.55) -- (432,41.55) .. controls (436.42,41.55) and (440,45.13) .. (440,49.55) -- (440,73.55) .. controls (440,77.96) and (436.42,81.55) .. (432,81.55) -- (105.18,81.55) .. controls (100.76,81.55) and (97.18,77.96) .. (97.18,73.55) -- cycle ;

\draw (41.54,54.38) node [anchor=north west][inner sep=0.75pt]  [font=\large] [align=left] {{\footnotesize Rack 1}};
\draw (209.31,144.65) node [anchor=north west][inner sep=0.75pt]  [font=\large,rotate=-90.04]  {$...$};
\draw (53.54,109.22) node [anchor=north west][inner sep=0.75pt]  [font=\large] [align=left] {{\footnotesize Rack 2}};
\draw (50.99,194.1) node [anchor=north west][inner sep=0.75pt]  [font=\large] [align=left] {{\footnotesize Rack $\displaystyle \bar{n}$}};
\draw (95.53,143.64) node [anchor=north west][inner sep=0.75pt]  [font=\large,rotate=-90.04]  {$...$};
\draw (119.86,53.43) node [anchor=north west][inner sep=0.75pt]  [font=\footnotesize] [align=left] {$\displaystyle c_{1,1}$};
\draw (172.78,53.43) node [anchor=north west][inner sep=0.75pt]  [font=\footnotesize] [align=left] {$\displaystyle c_{1,2}$};
\draw (277.61,53.43) node [anchor=north west][inner sep=0.75pt]  [font=\footnotesize,color={rgb, 255:red, 208; green, 2; blue, 27 }  ,opacity=1 ] [align=left] {$\displaystyle c_{1,u}$};
\draw (120.86,107.33) node [anchor=north west][inner sep=0.75pt]  [font=\footnotesize] [align=left] {$\displaystyle c_{2,1}$};
\draw (120.86,192.16) node [anchor=north west][inner sep=0.75pt]  [font=\footnotesize] [align=left] {$\displaystyle c_{\bar{n} ,1}$};
\draw (172.78,107.33) node [anchor=north west][inner sep=0.75pt]  [font=\footnotesize] [align=left] {$\displaystyle c_{2,2}$};
\draw (277.61,107.33) node [anchor=north west][inner sep=0.75pt]  [font=\footnotesize,color={rgb, 255:red, 208; green, 2; blue, 27 }  ,opacity=1 ] [align=left] {$\displaystyle c_{2,u}$};
\draw (172.78,192.16) node [anchor=north west][inner sep=0.75pt]  [font=\footnotesize] [align=left] {$\displaystyle c_{\bar{n} ,2}$};
\draw (275.61,193.16) node [anchor=north west][inner sep=0.75pt]  [font=\footnotesize,color={rgb, 255:red, 208; green, 2; blue, 27 }  ,opacity=1 ] [align=left] {$\displaystyle c_{\bar{n} ,u}$};
\draw (226.73,61.31) node [anchor=north west][inner sep=0.75pt]  [font=\large,rotate=-359.95]  {$...$};
\draw (227.73,114.21) node [anchor=north west][inner sep=0.75pt]  [font=\large,rotate=-359.95]  {$...$};
\draw (227.73,199.04) node [anchor=north west][inner sep=0.75pt]  [font=\large,rotate=-359.95]  {$...$};
\draw (393.61,54.16) node [anchor=north west][inner sep=0.75pt]  [font=\footnotesize] [align=left] {$\displaystyle c_{1,1}$};
\draw (374.31,124.65) node [anchor=north west][inner sep=0.75pt]  [font=\large,rotate=-90.04]  {$...$};
\draw (399.61,112.16) node [anchor=north west][inner sep=0.75pt]  [font=\footnotesize] [align=left] {$\displaystyle \bar{d}$};
\draw (329.61,90.72) node [anchor=north west][inner sep=0.75pt]  [font=\footnotesize] [align=left] {$\displaystyle \bar{\beta }_{1}$};
\draw (331.61,165) node [anchor=north west][inner sep=0.75pt]  [font=\footnotesize] [align=left] {$\displaystyle \bar{\beta }_{\bar{d}}$};
\draw (258.61,3.72) node [anchor=north west][inner sep=0.75pt]  [font=\footnotesize] [align=left] {$\displaystyle \alpha _{2}$};
\draw (285.61,24) node [anchor=north west][inner sep=0.75pt]  [font=\footnotesize] [align=left] {$\displaystyle \alpha _{u}$};
\end{tikzpicture}
\end{center}
\vspace{-0.3cm}
\caption{Rack-Aware Storage Model}
\end{figure}

Due to the feature of the rack-aware system,  each rack can be regarded as a cohesive unit for data transmission. Therefore, similar to distinct nodes in the homogeneous model, different racks in the rack-aware storage model can be considered homogeneous. In this paper, we focus on the so-called {\textbf{rack-aware system with locality}} \cite{CMS+23}. As presented in Fig. $3$, suppose that each rack corresponds to a repair set with $(r=u-1,\delta=2)$-locality, then the single failure $c_{2,1}$ in Rack $2$ can be recovered by downloading $\alpha'_{2,2},\cdots,\alpha'_{2,u}$ symbols from the internal surviving nodes. For the Rack $1$ where the number of failures exceeds the code locality, similar to previous discussion, $\bar{d}$ helper racks transmit $\bar{\beta_i}$, $i\in[\bar{d}]$ symbols to Rack $1$ by their relayers respectively, and the remaining nodes in this rack transmit $\alpha_{1,3},\cdots,\alpha_{1,u}$ symbols to the new replacement nodes. The repair bandwidth only counts as the inter-rack transmission, which is equal to $\bar{b}=\sum_{i=1}^{\bar{d}}\bar{\beta}_{i}$.

Specifically, we consider a locally repairable array code with $(r,\delta)$-locality such that
repair sets $\{S_{i_1},S_{i_2},\cdots,S_{i_{\bar{n}}}\}$ forms a partition of $[n]$.
By setting each repair set as a rack, we obtain a code with locality in each rack.
For locally repairable codes under the rack-aware storage model with locality, we have the following cut-set bound.

\begin{thm}[\cite{CMS+23}]\label{thm: cut-set bound}
Let $\mathcal{C}$ be a $(u\bar{n},u\bar{k};l)$ locally repairable array code with $(r,\delta=u-r+1)$-locality, where $u$ be the size of the local repair group. Denote by $\e_i$ the number of failures in $i$-th group. For any $i\in[\bar{n}]$, $\e_{i} \in [u]$ and any subset $D\subseteq[\bar{n}]\backslash \{i\}$ with $|D|=d\, (\bar{k}\leq d\leq\bar{n}-1)$, the bandwidth satisfies
\begin{equation*}
\mathcal{B}(\mathcal{C},\{i\},\e_{i},D)\geq
\begin{cases}
\frac{d(\e_{i}-\delta+1)l}{d-\bar{k}+1},& \e_{i}\geq \delta,\\
0,& otherwise.
\end{cases}
\end{equation*}

\end{thm}

\vspace{-0.5cm}

\tikzset{every picture/.style={line width=0.75pt}} 
\begin{figure}[ht]\label{fig: rack_lrc}  \begin{center}

\begin{tikzpicture}[x=0.5pt,y=0.5pt,yscale=-1,xscale=1]

\draw   (101.11,295.14) -- (307.58,295.14) -- (307.58,322.75) -- (101.11,322.75) -- cycle ;
\draw    (153.25,295.14) -- (153.25,322.42) ;
\draw    (205.39,295.14) -- (205.39,322.42) ;
\draw    (257.54,295.14) -- (257.54,322.42) ;
\draw   (101.11,349.7) -- (307.58,349.7) -- (307.58,377.44) -- (101.11,377.44) -- cycle ;
\draw    (153.25,349.7) -- (153.25,376.99) ;
\draw    (205.39,349.7) -- (205.39,376.99) ;
\draw    (257.54,349.7) -- (257.54,376.99) ;
\draw   (101.11,434.22) -- (307.58,434.22) -- (307.58,461.84) -- (101.11,461.84) -- cycle ;
\draw    (153.25,434.22) -- (153.25,461.5) ;
\draw    (205.39,434.22) -- (205.39,461.5) ;
\draw    (257.54,434.22) -- (257.54,461.5) ;
\draw [color={rgb, 255:red, 208; green, 2; blue, 27 }  ,draw opacity=1 ]   (101.11,322.75) -- (153.25,295.47) ;
\draw [color={rgb, 255:red, 208; green, 2; blue, 27 }  ,draw opacity=1 ]   (101.11,295.47) -- (153.25,322.75) ;
\draw   (350,295.22) -- (400.58,295.22) -- (400.58,322.84) -- (350,322.84) -- cycle ;
\draw  [dash pattern={on 4.5pt off 4.5pt}]  (284,294.57) .. controls (323.2,265.17) and (376.81,272.26) .. (399.64,294.2) ;
\draw [shift={(401,295.57)}, rotate = 226.27] [fill={rgb, 255:red, 0; green, 0; blue, 0 }  ][line width=0.08]  [draw opacity=0] (12,-3) -- (0,0) -- (12,3) -- cycle    ;
\draw    (308,362.5) .. controls (314.93,319.93) and (375.77,364.59) .. (399.86,324.09) ;
\draw [shift={(400.58,322.84)}, rotate = 118.93] [fill={rgb, 255:red, 0; green, 0; blue, 0 }  ][line width=0.08]  [draw opacity=0] (12,-3) -- (0,0) -- (12,3) -- cycle    ;
\draw    (307.58,447.84) .. controls (348.37,450.82) and (399.49,397.04) .. (400.57,323.94) ;
\draw [shift={(400.58,322.84)}, rotate = 90.45] [fill={rgb, 255:red, 0; green, 0; blue, 0 }  ][line width=0.08]  [draw opacity=0] (12,-3) -- (0,0) -- (12,3) -- cycle    ;
\draw  [dash pattern={on 4.5pt off 4.5pt}]  (349,349.5) .. controls (339.2,334.8) and (305.39,326.82) .. (285.22,348.22) ;
\draw [shift={(284,349.57)}, rotate = 310.93] [fill={rgb, 255:red, 0; green, 0; blue, 0 }  ][line width=0.08]  [draw opacity=0] (12,-3) -- (0,0) -- (12,3) -- cycle    ;
\draw  [dash pattern={on 4.5pt off 4.5pt}]  (180,349.57) .. controls (220.18,325) and (257.48,328.35) .. (282.48,348.32) ;
\draw [shift={(284,349.57)}, rotate = 220.12] [fill={rgb, 255:red, 0; green, 0; blue, 0 }  ][line width=0.08]  [draw opacity=0] (12,-3) -- (0,0) -- (12,3) -- cycle    ;
\draw  [dash pattern={on 4.5pt off 4.5pt}]  (180,433.57) .. controls (219.2,404.17) and (260.32,410.3) .. (282.66,432.21) ;
\draw [shift={(284,433.57)}, rotate = 226.27] [fill={rgb, 255:red, 0; green, 0; blue, 0 }  ][line width=0.08]  [draw opacity=0] (12,-3) -- (0,0) -- (12,3) -- cycle    ;
\draw  [dash pattern={on 4.5pt off 4.5pt}]  (129,433.57) .. controls (168.2,404.17) and (258.3,410.3) .. (282.6,432.21) ;
\draw [shift={(284,433.57)}, rotate = 226.27] [fill={rgb, 255:red, 0; green, 0; blue, 0 }  ][line width=0.08]  [draw opacity=0] (12,-3) -- (0,0) -- (12,3) -- cycle    ;
\draw    (380,331.07) .. controls (381,363.5) and (399,364.5) .. (417,365.5) ;
\draw [color={rgb, 255:red, 208; green, 2; blue, 27 }  ,draw opacity=1 ]   (153.25,322.75) -- (205.39,295.47) ;
\draw [color={rgb, 255:red, 208; green, 2; blue, 27 }  ,draw opacity=1 ]   (101.11,376.99) -- (153.25,349.7) ;
\draw [color={rgb, 255:red, 208; green, 2; blue, 27 }  ,draw opacity=1 ]   (153.25,295.47) -- (205.39,322.75) ;
\draw [color={rgb, 255:red, 208; green, 2; blue, 27 }  ,draw opacity=1 ]   (101.11,349.7) -- (153.25,376.99) ;
\draw   (400.58,295.22) -- (451.16,295.22) -- (451.16,322.84) -- (400.58,322.84) -- cycle ;
\draw  [dash pattern={on 0.84pt off 2.51pt}] (93.18,296.55) .. controls (93.18,292.13) and (96.76,288.55) .. (101.18,288.55) -- (451,288.55) .. controls (455.42,288.55) and (459,292.13) .. (459,296.55) -- (459,320.55) .. controls (459,324.96) and (455.42,328.55) .. (451,328.55) -- (101.18,328.55) .. controls (96.76,328.55) and (93.18,324.96) .. (93.18,320.55) -- cycle ;
\draw   (322,349.22) -- (372.58,349.22) -- (372.58,376.84) -- (322,376.84) -- cycle ;
\draw  [dash pattern={on 4.5pt off 4.5pt}]  (182,377.57) .. controls (228.3,401.14) and (318.25,399.55) .. (347.7,377.52) ;
\draw [shift={(349,376.5)}, rotate = 140.6] [fill={rgb, 255:red, 0; green, 0; blue, 0 }  ][line width=0.08]  [draw opacity=0] (12,-3) -- (0,0) -- (12,3) -- cycle    ;
\draw  [dash pattern={on 4.5pt off 4.5pt}]  (283,377.57) .. controls (310.58,390.31) and (319.73,397.29) .. (347.71,377.43) ;
\draw [shift={(349,376.5)}, rotate = 144.09] [fill={rgb, 255:red, 0; green, 0; blue, 0 }  ][line width=0.08]  [draw opacity=0] (12,-3) -- (0,0) -- (12,3) -- cycle    ;
\draw  [dash pattern={on 0.84pt off 2.51pt}] (91.18,352.07) .. controls (91.18,347.65) and (94.76,344.07) .. (99.18,344.07) -- (372,344.07) .. controls (376.42,344.07) and (380,347.65) .. (380,352.07) -- (380,376.07) .. controls (380,380.48) and (376.42,384.07) .. (372,384.07) -- (99.18,384.07) .. controls (94.76,384.07) and (91.18,380.48) .. (91.18,376.07) -- cycle ;

\draw (99.87,459.08) node [anchor=north west][inner sep=0.75pt]   [align=left] {$ $};
\draw (39.54,301.95) node [anchor=north west][inner sep=0.75pt]  [font=\large] [align=left] {{\footnotesize Rack 1}};
\draw (204.31,392.21) node [anchor=north west][inner sep=0.75pt]  [font=\large,rotate=-90.04]  {$...$};
\draw (38.54,350.78) node [anchor=north west][inner sep=0.75pt]  [font=\large] [align=left] {{\footnotesize Rack 2}};
\draw (45.99,441.66) node [anchor=north west][inner sep=0.75pt]  [font=\large] [align=left] {{\footnotesize Rack $\displaystyle \bar{n}$}};
\draw (90.53,391.21) node [anchor=north west][inner sep=0.75pt]  [font=\large,rotate=-90.04]  {$...$};
\draw (114.86,301) node [anchor=north west][inner sep=0.75pt]  [font=\footnotesize] [align=left] {$\displaystyle c_{1,1}$};
\draw (167.78,301) node [anchor=north west][inner sep=0.75pt]  [font=\footnotesize] [align=left] {$\displaystyle c_{1,2}$};
\draw (272.61,301) node [anchor=north west][inner sep=0.75pt]  [font=\footnotesize,color={rgb, 255:red, 208; green, 2; blue, 27 }  ,opacity=1 ] [align=left] {$\displaystyle c_{1,u}$};
\draw (115.86,354.89) node [anchor=north west][inner sep=0.75pt]  [font=\footnotesize] [align=left] {$\displaystyle c_{2,1}$};
\draw (115.86,439.73) node [anchor=north west][inner sep=0.75pt]  [font=\footnotesize] [align=left] {$\displaystyle c_{\bar{n} ,1}$};
\draw (167.78,354.89) node [anchor=north west][inner sep=0.75pt]  [font=\footnotesize] [align=left] {$\displaystyle c_{2,2}$};
\draw (272.61,354.89) node [anchor=north west][inner sep=0.75pt]  [font=\footnotesize,color={rgb, 255:red, 208; green, 2; blue, 27 }  ,opacity=1 ] [align=left] {$\displaystyle c_{2,u}$};
\draw (167.78,439.73) node [anchor=north west][inner sep=0.75pt]  [font=\footnotesize] [align=left] {$\displaystyle c_{\bar{n} ,2}$};
\draw (270.61,440.72) node [anchor=north west][inner sep=0.75pt]  [font=\footnotesize,color={rgb, 255:red, 208; green, 2; blue, 27 }  ,opacity=1 ] [align=left] {$\displaystyle c_{\bar{n} ,u}$};
\draw (221.73,308.88) node [anchor=north west][inner sep=0.75pt]  [font=\large,rotate=-359.95]  {$...$};
\draw (222.73,361.77) node [anchor=north west][inner sep=0.75pt]  [font=\large,rotate=-359.95]  {$...$};
\draw (222.73,446.61) node [anchor=north west][inner sep=0.75pt]  [font=\large,rotate=-359.95]  {$...$};
\draw (363.61,301.72) node [anchor=north west][inner sep=0.75pt]  [font=\footnotesize] [align=left] {$\displaystyle c_{1,1}$};
\draw (409.61,345.72) node [anchor=north west][inner sep=0.75pt]  [font=\footnotesize] [align=left] {$\displaystyle \bar{d}$};
\draw (414.61,301.72) node [anchor=north west][inner sep=0.75pt]  [font=\footnotesize] [align=left] {$\displaystyle c_{1,2}$};
\draw (380,331.07) node [anchor=north west][inner sep=0.75pt]  [font=\footnotesize] [align=left] {$\displaystyle \bar{\beta }_{1}$};
\draw (374.61,403.72) node [anchor=north west][inner sep=0.75pt]  [font=\footnotesize] [align=left] {$\displaystyle \bar{\beta }_{\bar{d}}$};
\draw (335.61,355.72) node [anchor=north west][inner sep=0.75pt]  [font=\footnotesize] [align=left] {$\displaystyle c_{2,1}$};
\draw (330.61,255) node [anchor=north west][inner sep=0.75pt]  [font=\footnotesize] [align=left] {$\displaystyle \alpha _{1,u}$};
\draw (221.61,332) node [anchor=north west][inner sep=0.75pt]  [font=\footnotesize] [align=left] {$\displaystyle \alpha _{2,2}$};
\draw (333.61,328) node [anchor=north west][inner sep=0.75pt]  [font=\footnotesize] [align=left] {$\displaystyle \alpha _{2,1}$};
\draw (240.61,393) node [anchor=north west][inner sep=0.75pt]  [font=\footnotesize] [align=left] {$\displaystyle \alpha '_{2,2}$};
\draw (314.61,388) node [anchor=north west][inner sep=0.75pt]  [font=\footnotesize] [align=left] {$\displaystyle \alpha '_{2,u}$};

\end{tikzpicture}
\end{center}
\vspace{-0.3cm}
\caption{Rack-Aware System with Locality}
\end{figure}
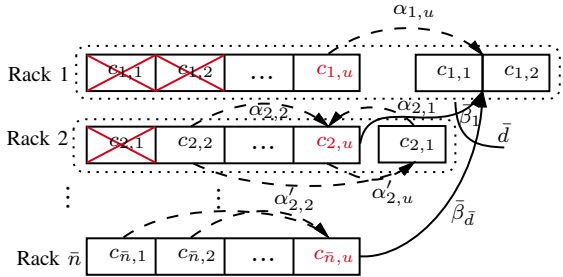

\section{Repair Scheme for Rack-Aware Locally Repairable Codes}
In this section, inspired by \cite{TB14},
we employ good polynomials to construct locally repairable array codes under the rack-aware storage model, such that all nodes in each rack form a repair set. It is an extension of the well-known Tamo-Barg code~\cite{TB14} to the array form. Combining with the technique of regenerating codes, we design a generic repair scheme to handle multiple-node failures of the proposed construction beyond the code locality.
To begin with, we recall some definitions and properties of good polynomials.

\begin{defn}[Good Polynomial \cite{TB14}]  A polynomial $h(x)\in \fq[x]$ of degree $u$ is called a good polynomial if there exists a partition $A=\cup_{i=1}^{\bar{n}}A_{i}$ over $\fq$ of size $n$ with $|A_{i}|=u$, such that $h(x)$ remains a constant on each set $A_{i}$. In other words, $h(\alpha)=y_{i}$ for any $\alpha\in A_{i}$, where $y_{i}\in\fq$ for any $i\in[\bar{n}]$.
\end{defn}

\begin{lem}[Chinese Remainder Theorem, \cite{MS77}]\label{crt}
Let $h_1(x),\cdots ,h_{s}(x)$ be
pairwise coprime polynomials over $\f_{q}$ of degree $u$.
For any  $s$ polynomials $f_{i}(x)\in \f_{q}[x]$, there exists a unique polynomial $f(x)$ of degree less than $su$ satisfying
$$f(x)\equiv f_i(x)\bmod\,h_{i}(x),\,i\in[s].$$
\end{lem}

\begin{lem}[\cite{CMS+23}]\label{cross-rack rs}Let $h(x)\in \fq[x]$ be a monic polynomial of degree $u$ and $0<r\leq u$.
Suppose that $y_{1},y_{2},\cdots,y_{s}\in \fq$ are  $s$ distinct constants such that
\begin{equation}\label{eq: crt}
f(x) \,\bmod\,h(x)-y_{i}\equiv f_{i}(x)=\sum_{j=0}^{r-1}e_{i,j}x^j,\,\,\,\,i\in[s].
\end{equation}
 Then for any $y\in \f_{q}$,
 \begin{equation}\label{eq: residue}
 f(x)\bmod\,h(x)-y\equiv\sum_{j=0}^{r-1}H_{j}(y)x^{j},
 \end{equation}
where $deg(H_{j}(x))\leq s-1$  and $H_{j}(y_{i})=e_{i,j}$ for $i\in[s], j\in[0,r-1].$
\end{lem}

\subsection{A Generic Construction}\label{sec: GC}

Let $h(x)$ be a good polynomial of degree $u$. For any $w\in [0,l-1]$, the $\bar{n}$ distinct elements
 $y^{(w)}_1,\cdots, y^{(w)}_{\bar{n}}$ over $\f_{q}$  satisfy
$h(\alpha)=y^{(w)}_i$ for $\alpha\in A^{(w)}_i$, where $|A^{(w)}_i|=u$ for $ i\in [\bar{n}]$.
Consider a vector ${\bm f}({\bm X})=(f^{(0)}(x_0),\cdots,f^{(l-1)}(x_{l-1}))^\top$, where
 $f^{(i)}(x_i)\in \fq[x]$ and $\bm X=(x_0,x_1,\ldots,x_{l-1})^\top$. Suppose that each component of ${\bm f}({\bm X})$ satisfies
\begin{equation}\label{eq: def_poly}
f^{(w)}(x) \equiv f^{(w)}_{i}(x)\,\bmod\,h(x)-y^{(w)}_{i},\,\,\,i\in[\bar{k}],
\end{equation}
where $w\in[0,l-1]$, $\bar{k}<\bar{n}$ and $deg(f^{(w)}_{i}(x))\leq r-1$.

Define an array code
\begin{equation}\label{eq: rack_lrc}
\mathcal{C}=\{{\bm f}({\bm{\alpha}_{i,j}}): i\in[\bar{n}],j\in[u]\},
\end{equation}
where ${\bm f}({\bm X})=(f^{(0)}(x_0),\cdots,f^{(l-1)}(x_{l-1}))^\top$ satisfies \eqref{eq: def_poly} and 
\begin{equation*}
\bm{\alpha}_{i,j}=(\alpha^{(0)}_{i,j},\alpha^{(1)}_{i,j},\cdots,\alpha^{(l-1)}_{i,j})\in A^{(0)}_{i}\times\cdots\times A^{(l-1)}_{i}.
\end{equation*}
Then, arrange these nodes corresponding to evaluation points $\{\bm\alpha_{i,j}:j\in[u]\}$ into $i$-th rack. We present it as an array
 \begin{equation}\label{eq: lrc_array}
 \begin{pmatrix}
 {\bm f}(\bm\alpha_{1,1})& {\bm f}(\bm\alpha_{1,2})&\cdots& {\bm f}(\bm\alpha_{1,u})\\
  {\bm f}(\bm\alpha_{2,1})& {\bm f}(\bm\alpha_{2,2})&\cdots& {\bm f}(\bm\alpha_{2,u})\\
  \vdots&\vdots&\cdots&\vdots\\
   {\bm f}(\bm\alpha_{\bar{n},1})& {\bm f}(\bm\alpha_{\bar{n},2})&\cdots& {\bm f}(\bm\alpha_{\bar{n},u})\\
 \end{pmatrix},
 \end{equation} which implies that each row
$
\{{\bm f}(\bm{\alpha}_{i,j}): j\in[u]\}
$
are nodes in Rack $i$ and each node stores an $l$-length vector.

\begin{thm}\label{thm:rack_rs}{Let $0< r <u$. For any $i\in [\bar{n}]$, $w\in [0,l-1]$, $j_1\ne j_2\in [u]$,
if $\alpha^{(w)}_{i,j_1}\ne \alpha^{(w)}_{i,j_2}$, then the code $\mathcal{C}$ of \eqref{eq: rack_lrc} is an $(n=u\bar{n},r\bar{k};l)$ locally repairable array code with $(r,u-r+1)$-locality.} Furthermore, the residue polynomials of
\begin{equation}
f^{(w)}(x)\,\bmod\,h(x)-y^{(w)}_{i},\,\,\,i\in[\bar{n}],\,w\in[0,l-1]
\end{equation}
can be represented as
$
f^{(w)}_i(x)=\sum_{j=0}^{r-1}e^{(w)}_{i,j}x^{j}.
$
For any $j\in[0,r-1]$,
\begin{equation}\label{eq: cross_rack_mds}
(\bm{e}_{1,j},\bm{e}_{2,j},\cdots,\bm{e}_{\bar{n},j})\,\,
\end{equation}
are codewords of an $(\bar{n},\bar{k};l)$ MDS array code, where $\bm{e}_{i,j}=(e^{(0)}_{i,j},e^{(1)}_{i,j},\cdots,e^{(l-1)}_{i,j})^{\top}$.
\end{thm}

\begin{proof}
For any $w\in[0,l-1]$, according to Lemma~\ref{cross-rack rs}, we can claim that $f^{(w)}_{i}(x)$ is of degree at most $r-1$ for $i\in[\bar{n}]$, thus can be represented as
$
f^{(w)}_i(x)=\sum_{j=0}^{r-1}e^{(w)}_{i,j}x^{j}.
$
For any $i\in [\bar{n}]$, $w\in [0,l-1]$, {we have
$\alpha^{(w)}_{i,j_1}\ne \alpha^{(w)}_{i,j_2}$ for $j_1\ne j_2\in [u]$ and $f^{(w)}(\alpha)=f^{(w)}_{i}(\alpha)$ for $\alpha\in A^{(w)}_{i}$.}
Thus, the vector $
(f^{(w)}(\alpha^{(w)}_{i,1}),f^{(w)}(\alpha^{(w)}_{i,2}),\cdots,f^{(w)}(\alpha^{(w)}_{i,u}))$ with $A^{(w)}_{i}=\{\alpha^{(w)}_{i,1},\alpha^{(w)}_{i,2},\cdots,\alpha^{(w)}_{i,u}\}$
{forms a codeword of a $[u,r]$ RS with evaluation points $A^{(w)}_{i}$, thereby each row in \eqref{eq: lrc_array} is a codeword of a $(u,r,l)$ MDS array code,
which  forms a repair set of $(r,u-r+1)$ locality according to Definition \ref{def: array_lrc}.}

 Moreover, by Lemma \ref{cross-rack rs}, for $w\in[0,l-1]$, we have
\begin{equation*}
f^{(w)}_i(x)=\sum_{j=0}^{r-1}e^{(w)}_{i,j}x^{j}=\sum_{j=0}^{r-1}H^{(w)}_{j}(y^{(w)}_i)x^{j},\,\,i\in[\bar{n}],
\end{equation*}
where $H^{(w)}_j(x)$ is a polynomial of degree less than $\bar{k}$.
Hence, $(e^{(w)}_{1,j},e^{(w)}_{2,j},\cdots,e^{(w)}_{\bar{n},j})$
is a codeword of an $[\bar{n},\bar{k}]$ RS code for
any $j\in[0,r-1]$, which implies that
$(\bm{e}_{1,j},\bm{e}_{2,j},\cdots,\bm{e}_{\bar{n},j}),$ $ j\in[0,r-1]$
are codewords of an $(\bar{n},\bar{k},l)$ MDS array code. \end{proof}



\subsection{Repair Mechanism}\label{sec: repair_mechanism}

In this subsection, we discuss the repair mechanism of the generic construction in Section~\ref{sec: GC}. As mentioned earlier, the array code $\mathcal{C}$ exhibits $(r, u-r+1)$-locality within each rack. This implies that if the number of failed nodes within racks is less than $u-r+1$, the local repair group can successfully recover these failures solely from internal nodes. However, if the number of failures exceeds the locality, we employ the properties between the racks to handle the ``extra'' failures. We provide a detailed illustration of various erasure patterns and their corresponding repair mechanisms below.

Denote by $R=\{i_1,\cdots,i_m\}\subset [\bar{n}]$ the set of racks contained failures, and let $\e_{i_{j}}$ be the number of failed nodes in $i_j$-th rack. Due to the locality, we classify the rack failure into two scenarios: one that can be recovered by the internal nodes within the same rack, and the other that requires the help of other surviving racks. Let  $m_1,m_2$ be the number of racks containing failed nodes corresponding to these two scenarios respectively, i.e.,
\begin{equation*}
\begin{cases}
  1\leq \e_{i_{j}}\leq u-r&j\in[m_{1}], \\
   u-r<\e_{i_{j}}\leq u& j\in[m_1+1,m_1+m_2],
\end{cases}
\end{equation*}
 where $|R|=m=m_{1}+m_{2}$, $m_2\leq \bar{n}-\bar{k}$. Based on this, we have the following erasure patterns.

{\textbf{Case~1:} $m_{1}=m$,  $m_{2}=0$: Since the nodes in each rack form a codeword of a $[u,r]$ RS code, for $j\in[m]$, the rack $i_{j}\in R$  downloads all symbols from $r$ remaining nodes within itself to independently repair at most $u-r$ failures.

{\textbf{Case~2:} $m_1=0$,  $m_{2}=m$:} Let $\e'_{i_{j}}=\e_{i_{j}}-u+r$ for $j\in[m]$. For $i_{j}\in R$, Rack $i_j$ computes $\e'_{i_{j}}$ coefficients $\{\bm{e}_{i_{j},u-\tau}:\tau\in[\e'_{i_{j}}]\}$ of the residue polynomial $\bm{f}_{i_{j}}(\bm{X})$. Then, by combining data from the $u-\e_{i_{j}}$ surviving nodes in this rack,  each component of $\bm{f}_{i_{j}}(\bm{X})$ with degree $r-1$ can be determined, thereby recovering all the failures.

{\textbf{Case~3:} $m_{1},m_{2}>0$:} The hybrid erasure pattern can be decomposed into parallel repair processes corresponding to {\textbf{Case $1$}} and $2$.
For the racks contained the number of failures less than $u-r+1$, i.e. Racks $i_{j}$ for $j\in[m_1]$, conduct the repair of {\textbf{Case} $1$}. Otherwise, the repair is in accordance with {\textbf{Case} $2$}.

Notably, in {\textbf{Case} $3$},  since the $(r,u-r+1)$-locality of each rack, $i_j$-th rack for $j\in[m_1]$ suffices to compute its residue polynomial $\bm{f}_{i_j}(\bm{X})$ by connecting at least $r$ surviving nodes.  These racks have no effect on the repair of racks containing more than $u-r$ failures. For the sake of simplicity, we omit the procedure of local repair and only discuss how to repair the racks that occur more than $u-r$ failures, i.e., the repair of {\textbf{Case} $2$}.

It is clear that each component of the code $\mathcal{C}$ in \eqref{eq: rack_lrc} can be seen as a subcode of a rack-aware RS code given in \cite{YCT23}.  Motivated by this, we generalize its repair framework to adapt array codes that have locality in each rack.
Assume that there are $m$ racks suffering failures indexed by the set $R=\{i^*_{1},\cdots,i^*_{m}\}\subseteq [\bar{n}]$.
 For $1\leq \tau\leq m$, $\e_\tau$ is denoted as the number of failures in Rack $i^*_{\tau}$.
 Let $D=\{i_{1},\cdots,i_{{d}}\}\subseteq [\bar{n}]\backslash\{i^*_{1},\cdots,i^*_{m}\}$ be the set of helper racks of size ${d}$.
 The following procedure shows the concrete repair of the rack-aware locally repairable code given in \eqref{eq: rack_lrc}.

 {
\textbf{Repair procedure of codewords by \eqref{eq: rack_lrc}:}
\begin{enumerate}
\item[Step 1.] For $i\in D$, Rack $i$ computes $\{f^{(w)}_{i}(x):w\in[0,l-1]\}$ from any $r$ storage nodes within the rack by Lagrange interpolation, thereby obtaining an array consisting of coefficients
$(\bm{e}_{i,0},\bm{e}_{i,1},\cdots,\bm{e}_{i,u-1}).$

\item[Step 2.] For any $t\in[r]$,
$
c_{r-t}=(\bm{e}_{1,r-t},\bm{e}_{2,r-t},\cdots,\bm{e}_{\bar{n},r-t}),
$
forms a codeword of an $(\bar{n},\bar{k},l)$ MDS array code.

 Let $\e'_{\tau}=\epsilon_\tau-u+r$ and $\epsilon=\max_{1\leq \tau\leq m}\epsilon'_\tau$.
Define
\begin{equation*}\label{eqn_def_Ri}
R_t\triangleq \{i^*_{\tau}~:~\epsilon'_{\tau}\geq t, 1\leq \tau\leq m\},\,\,t\in[\e],
\end{equation*}
 then recover the $|R_{t}|$ symbols
\begin{eqnarray}\label{eq_co}
\{\bm{e}_{i, r-t}: \,i\in R_t\}
\end{eqnarray}
of the MDS array codeword $c_{r-t}$.

\item[Step 3.] Given $i_{\tau}^*$ and $\tau\in [m]$, let the coordinates of $u-\epsilon_{\tau}$ surviving nodes be
$j_1,\cdots,j_{u-\epsilon_{\tau}}$. Then, for each $i_{\tau}^*$ and $w\in[0,l-1]$,  by \eqref{eq_co}, one can compute the polynomial
\begin{equation*}
f_{i^*_{\tau}}^{(w)}(x)-\sum_{t=1}^{\e_{\tau}-u+r} e^{(w)}_{i^*_{\tau},r-t}x^{r-t},
\end{equation*}
whose degree is at most $u-\e_{\tau}-1$, from the  $u-\epsilon_{\tau}$ surviving nodes $\{\bm{f}(\bm{\alpha}_{i_{\tau}^*,j_1}),\cdots,\bm{f}(\bm{\alpha}_{i_{\tau}^*,j_{u-\e_{\tau}}})\}$  in rack $i_{\tau}^*$. Then, the repair center can compute the remainder coefficients $\{\bm{e}_{i_{\tau}^*, r-t}: \,t\in[\e'_{\tau}+1,r]\}$
and thereby the entire polynomials
\begin{eqnarray*}
f^{(w)}_{i^*_{\tau}}(x)=\sum_{j=0}^{r-1} e^{(w)}_{i^*_{\tau},j}x^j, \,\, w\in[0,l-1],
\end{eqnarray*}
which can  figure out the $\e_{\tau}$ failures
\begin{equation*}
\{\bm{f}(\bm{\alpha}_{i^*_{\tau},j_{u-\e_{\tau}+1}}),\cdots,\bm{f}(\bm{\alpha}_{i^*_{\tau},j_{{u}}})\}.
\end{equation*}

\end{enumerate}

According to the assumption of the rack-aware storage model, the repair bandwidth is only required to count the data transmission across the rack to repair the symbols of \eqref{eq_co} in Step 2, i.e., the amount of data transmitted from helper racks (repair sets) to the rack that contains the number of failures more than $u-r$.

\begin{thm}\label{thm_general_repair} For a rack-aware locally repairable array code defined by \eqref{eq: rack_lrc},
 assume that the $m$ racks $\{i_\tau^*:\tau\in[m]\}\subseteq [\bar{n}]$
contain $\{\e_\tau:\tau\in[m]\}$ failures respectively. Setting $\epsilon=\max_{1\leq \tau\leq m}\epsilon'_\tau$, where $\e'_{\tau}=\epsilon_\tau-u+r$ and $\epsilon_\tau>u-r$.
For $t\in[ \epsilon]$,  if the repair bandwidth $b_{|R_t|}$
 is capable to recover  the coefficients  in  \eqref{eq_co}  from $|D|$ helper racks, the failures can be recovered with
  bandwidth
$
 b=\sum_{1\leq t\leq \e}b_{|R_t|}.
$
 \end{thm}}

\begin{proof}
According to Step $3$ of \textbf{Repair procedure of codewords by \eqref{eq: rack_lrc}}, the residue polynomial $\bm{f}_{i^*_{\tau}}(\bm{X})$ for $\tau\in[m]$ can be computed from \eqref{eq_co} and the surviving nodes of Rack $i^*_{\tau}$, which means that all of the nodes in failed racks can be determined. Since the assumption of the rack-aware model, only the cross-rack transmission is taken into account of the bandwidth. Thus the repair bandwidth is exactly $b=\sum_{1\leq t\leq \e}b_{|R_t|}.$
\end{proof}

\section{A Family of Bandwidth-Optimal Locally Repairable Array Codes}
In this section, we consider the scenario that a rack, or a single local repair set suffers failures exceeding its local repair capacity, i.e., $m=1$. We present an explicit construction of locally repairable array codes with desirable repair properties based on the generic construction given in Section~\ref{sec: GC}.
 Our construction is partially inspired by the known MSR code proposed in literature \cite{YB17}.
 The resulting codes support the optimal repair of a local repair set (or rack) that contains failures  $\e_{\tau}>u-r$.



Let $\bar{s}={d}-\bar{k}+1$ and $l=\bar{s}^{\bar{n}}$, where ${d}$ is the number of helper racks.  Suppose that $|\fq|>\bar{s}n$ and let $\lambda\in\fq$ be an element of multiplicative order $\bar{s}n$.  For an integer $w\in[0,l-1]$, denote $w_{i}$ as the $i$-th coordinate of the $\bar{s}$-ary expansion of $w$ and
define the set of evaluation points $A^{(w)}=\cup_{i=1}^{\bar{n}}A_{i}^{(w)}=\{\lambda_{ij}^{(w)}:i\in[\bar{n}],j\in[u]\}$, where $\lambda^{(w)}_{ij}=\lambda^{(i-1)\bar{s}+w_{i}}(\lambda^{\bar{s}\bar{n}})^{j-1}$, for simplicity, we denote $\lambda_{i,w}=\lambda^{(i-1)\bar{s}+w_{i}}$.

\begin{cons}\label{cons: single-rack}Let $h(x)=x^u$ and $0< r <u.$
For positive integer $\bar{k}<\bar{n}$ and $w\in[0,l-1]$, let
\begin{equation*}
\begin{split}
\Psi^{(w)}(\bar{k},u,&r)\triangleq \{f(x):\deg(f(x))<u(\bar{k}-1)+r, \\
&f(x) \equiv f_{i}(x)\,\bmod\,h(x)-\lambda_{i,w}^u,
 \deg(f_i(x))<r
\}
\end{split}
\end{equation*}
For $\bm f=(f^{(0)}(x),f^{(1)}(x),\cdots,f^{(l-1)}(x))^\top$ and $\bm X=(x_0, x_1,\cdots, x_{l-1})^\top\in \fq^{l}$, let the polynomial
${\bm f}({\bm X})\triangleq (f^{(0)}(x_0),\cdots,f^{(l-1)}(x_{l-1}))^\top$.

Define an array code
\begin{equation}\label{eq: rack_lrc_single}
\begin{split}
\mathcal{C}=\{{\bm f}({\bm\lambda}_{ij})_{i\in[\bar{n}],j\in[u]}~:~&{\bm f}=(f^{(0)}(x),f^{(1)}(x),\cdots,\\&f^{(l-1)}(x))^\top,\\
 &f^{(w)}(x)\in \Psi^{(w)}(\bar{k},u,r),\\
 & \text{for }w\in [0,l-1]\},
\end{split}
\end{equation}
where
$
\bm{\lambda}_{ij}=(\lambda^{(0)}_{ij},\lambda^{(1)}_{ij},\cdots,\lambda^{(l-1)}_{ij})^\top.
$
Then, arrange the nodes corresponding to evaluation points $\{\bm\lambda_{ij}:j\in[u]\}$ into $i$-th rack.
\end{cons}

 The repair procedure follows from Section~\ref{sec: repair_mechanism}.
Since Steps 1 and 3 are the same for different constructions, herein we only discuss Step 2 in detail.

\begin{thm} The code $\mathcal{C}$ given in Construction~\ref{cons: single-rack} is a rack-based locally repairable array code with optimal repair bandwidth for the repair of any number of failures in a single rack from any $d$ helper racks.
\end{thm}

\begin{proof}
Clearly, $h(x)=x^u$ is a good polynomial satisfying
 $h(\alpha)=\lambda_{i,w}^u$ for $\alpha\in A_{i}^{(w)}$, $i\in[\bar{n}]$ and $w\in[0,l-1]$. Denote $y^{(w)}_{i}=\lambda_{i,w}^u$. Since $0<r<u$, by Theorem \ref{thm:rack_rs}, $\mathcal{C}$ is an $(n,r\bar{k};l=\bar{s}^{\bar{n}})$ locally repairable array code.

Let the original codeword $C={\bm f}({\bm \lambda}_{ij})$ determined by $\bm f=(f^{(0)}(x),f^{(1)}(x),\cdots,f^{(l-1)}(x))^\top$
with $f^{(w)}\in \Psi^{(w)}(\bar{k},u,r)$ for $w\in [0,l-1]$.
Thus, according to Theorem~\ref{thm:rack_rs}, for any $w\in[0,l-1]$, denote by the residual polynomials $f^{(w)}_i$ with $\deg(f^{(w)}_i)<r$,
i.e.,
$
f^{(w)}_{i}(x)=\sum_{j=0}^{r-1}e^{(w)}_{i,j}x^j,\,\,\, i\in[\bar{n}],
$
then $(e^{(w)}_{1,j},e^{(w)}_{2,j},\cdots,e^{(w)}_{\bar{n},j})$
forms a codeword of an $[\bar{n},\bar{k}]$ RS code with evaluation points $\bar{A}^{(w)}=\{\lambda_{i,w}^u:i\in[\bar{n}]\}$  for
any $j\in[0,r-1]$.
 Therefore, one can deduce the following parity check equations:
\begin{equation}\label{eq: pc1}
\sum_{i=1}^{\bar{n}}{\lambda_{i,w}^{u{t}}}e^{(w)}_{i,j}=\sum_{i=1}^{\bar{n}}\lambda^{((i-1)\bar{s}+w_{i})ut}e^{(w)}_{i,j}=0
\end{equation}
for $t\in[0,\bar{n}-\bar{k}-1]$, $w\in[0,l-1]$ and $j\in[0,r-1]$.

 Consider the repair of $\e_{\tau}$ failures in $\tau$-th rack. If $\e_{\tau}\leq u-r$, the failures can be recovered from the locality without counting into the repair bandwidth, thus we discuss the case that $u-r< \e_{\tau}\leq u$. By Step $2$ in \textbf{Repair procedure of \eqref{eq: rack_lrc}}, we have 
$
R_1=R_2=\cdots=R_{\e_{\tau}}=\{\tau\},
$
and our target is to repair coefficients
\begin{equation}\label{eq: target_coefficients}
\{e^{(w)}_{\tau,r-t}:t\in[\e_{\tau}']\},\,\,w\in[0,l-1],
\end{equation}
where $\e_{\tau}'=\e_{\tau}-u+r.$

For $j\in[r-\e'_{\tau},r-1],$ summing the equations~\eqref{eq: pc1} on $w_{\tau}\in[0,\bar{s}-1].$ Then, for $t\in[0,\bar{n}-\bar{k}-1]$, we can obtain equations
\begin{equation}\label{eq: pc2}
\sum_{w_{\tau}=0}^{\bar{s}-1}\lambda^{((\tau-1)\bar{s}+w_{\tau})ut}e^{(w)}_{\tau,j}=-\sum_{i\neq\tau}\lambda^{((i-1)\bar{s}+w_{i})ut}\sum_{w_{\tau}=0}^{\bar{s}-1}e^{(w)}_{i,j},
\end{equation}
which define a GRS code of length $\bar{s}+\bar{n}-1$ and dimension $\bar{s}+\bar{n}-1-(\bar{n}-\bar{k})={d}$, {where
$w_i$ denotes the $i$-th coordinate of the $\bar{s}$-ary expansion of $w$ and $w\in [0,\bar{s}^{\bar{n}}-1]$.} Hence, any ${d}$ helper racks suffice to determine the coefficients in \eqref{eq: target_coefficients}, thereby
 repairing the residue polynomial $\bm{f}_{\tau}(\bm{X})$ with the help of the reminder $u-\e_{\tau}$ surviving nodes in $\tau$-th rack. As a consequence, the $\e_{\tau}$ failures in rack $\tau$ can be recovered.

From equations \eqref{eq: pc2}, for each $j\in[r-\e'_{\tau},r-1]$, one needs to download $l/\bar{s}$ symbols from each helper rack. Thus, the total repair bandwidth is
\begin{equation*}
b=\frac{\e'_{\tau}{d}l}{\bar{s}}=\frac{(\e_{\tau}-u+r){d}l}{{d}-\bar{k}+1},
\end{equation*}
which meets the cut-set bound of Theorem~\ref{thm: cut-set bound}.
\end{proof}

\end{document}